\documentclass[sigconf,screen]{acmart}

\usepackage{algorithm}
\usepackage{algpseudocode}
\makeatletter
\renewcommand{\theHALG@line}{\thealgorithm.\arabic{ALG@line}}
\makeatother
\usepackage{cleveref}
\usepackage{enumitem}
\usepackage{microtype}
\usepackage{multirow}
\usepackage{siunitx}
\usepackage{tikz}
\usetikzlibrary{
  positioning,
  decorations.pathmorphing,
  fit,
  calc
}

\AtBeginDocument{%
  }

\setcopyright{cc}
\setcctype{by}
\acmDOI{10.1145/3832783.3834444}
\acmYear{2026}
\copyrightyear{2026}
\acmISBN{979-8-4007-2882-2/2026/10}
\acmConference[ASE '26]{Proceedings of the 41st IEEE/ACM International Conference on Automated Software Engineering}{October 12--16, 2026}{Munich, Germany}
\acmBooktitle{Proceedings of the 41st IEEE/ACM International Conference on Automated Software Engineering (ASE '26), October 12--16, 2026, Munich, Germany}
\acmSubmissionID{ase26main-p4025-p}
\received{2026-03-27}
\received[accepted]{2026-06-18}

\begin{document}

\title{GPU-Accelerated Belief Propagation for Program Analysis}

\author{Haoyu Feng}
\orcid{0009-0005-0012-3750}
\affiliation{%
  \department{Key Laboratory of High Confidence Software Technologies (Peking University), Ministry of Education; School of Computer Science}
  \institution{Peking University}
  \city{Beijing}
  \country{China}
}
\email{haoyufeng@stu.pku.edu.cn}

\author{Xin Zhang}
\orcid{0000-0002-1515-7145}
\affiliation{%
  \department{Key Laboratory of High Confidence Software Technologies (Peking University), Ministry of Education; School of Computer Science}
  \institution{Peking University}
  \city{Beijing}
  \country{China}
}
\email{xin@pku.edu.cn}
\authornote{Corresponding author.}

\renewcommand{\shortauthors}{Haoyu Feng and Xin Zhang}

\begin{abstract}
  Belief Propagation (BP) is a widely used approximate inference algorithm in probabilistic graphical models (PGMs), but it is computationally expensive when applied to large-scale program analysis. Existing GPU-based approaches are unable to support flexible update strategies and have yet to integrate logical constraints with GPU acceleration, leading to challenges in both generality and efficiency.
  
  We present \textsc{FastLBP}, a GPU-accelerated BP framework for program analysis. 
  We propose a unified representation for specifying flexible update strategies required in program analysis, along with a dependency analysis algorithm to enable parallel execution. 
  Furthermore, we implement BP with local structures on GPUs by assigning individual threads to message computations and utilizing a memory-efficient representation.
  Experiments on SmartFL and Bingo show that \textsc{FastLBP} achieves average speedups of $17.42\times$ and $2.82\times$ over CPU-based approaches on SmartFL and Bingo, respectively, and $6.14\times$ over the GPU-based approach on SmartFL, while preserving accuracy.
  Moreover, \textsc{FastLBP} supports update strategies that existing GPU-based approaches cannot support, demonstrating its improved generality for real-world program analysis.

\end{abstract}



%
% Keywords. The author(s) should pick words that accurately describe
% the work being presented. Separate the keywords with commas.
% % A "teaser" image appears between the author and affiliation
% % information and the body of the document, and typically spans the
% % page.


\maketitle

\section{Introduction}
\label{sec:intro}
Belief Propagation (BP) is a widely used approximate inference algorithm for probabilistic graphical models (PGMs)~\cite{pearl1988probabilistic}. In these models, random variables are represented as nodes, and their dependencies are captured by the graph structure. The primary task of inference is to estimate the posterior marginal distributions of these variables given observed evidence. Exact inference is often computationally intractable, making approximate algorithms a more practical alternative. BP is a prominent approximate inference algorithm~\cite{murphy1999loopy} that operates by iteratively passing messages between neighboring nodes until convergence or for a specified number of iterations. Due to its simplicity and efficiency, BP has been successfully applied in diverse fields including computer vision~\cite{felzenszawlb2006efficient}, channel coding~\cite{mackay1996near}, protein folding~\cite{yanover2002approximate}, and, notably, program analysis~\cite{raghothaman2018user}.

In recent years, the combination of program analysis and probabilistic reasoning has enabled the application of BP to program analysis~\cite{raghothaman2018user,heo2019continuously,chen2021boosting,zeng2022fault,kim2022learning,zhang2024learning,li2025combining,shi2025on,wu2025smartfl,zhang2026fuzzing, lin2026beer}. For instance, Bingo~\cite{raghothaman2018user} converts the Datalog program analysis into a directed acyclic graph, transforms the graph into a Bayesian network by attaching probabilities to logic rules, and applies BP to infer the likelihood of bug reports. SmartFL~\cite{zeng2022fault,wu2025smartfl} models program semantics as a factor graph by utilizing information from static analysis results and dynamic execution traces, and employs BP to infer the probabilities of program elements being faulty.

Program analysis on large-scale software often yields large PGMs, leading to high computational costs for BP. 
Reducing this inference cost is particularly important because probabilistic program analysis relies on BP to infer posterior probabilities conditioned on observations such as user feedback or dynamic execution results. Faster BP therefore reduces the time required to incorporate such observations and obtain posterior estimates.

Graphics Processing Units (GPUs), optimized for massive parallel operations, are a natural choice for reducing the inference cost of BP. This potential has motivated prior GPU-based BP approaches~\cite{brunton2006belief,grauergray2008gpu,merwe2019message,reddy2012a,abburi2011a,zhou2024pgmax}. 

However, existing GPU-based BP approaches suffer from two key limitations when applied to program analysis. 
First, they lack support for flexible message update strategies required in program analysis. In the existing BP literature, researchers have devoted considerable attention to update strategies as a means of improving BP performance. The most straightforward scheme is synchronous BP, where all messages are updated simultaneously~\cite{murphy1999loopy}. In contrast, asynchronous schemes introduce a degree of sequential updating, such as tree-based updates~\cite{wainwright2006tree} and greedy message selection~\cite{elidan2006residual,gonzalez2009residual,merwe2019message, aksenov2020relaxed}. While synchronous BP exposes high levels of parallelism and can approximate the correct posteriors well when it converges, asynchronous BP has been shown, both empirically and theoretically, to converge more reliably when synchronous BP fails to converge~\cite{elidan2006residual}. Program analysis applications often require different update strategies depending on usage scenarios, yet existing GPU-based implementations support only limited update strategies. For example, widely used tools such as PGMax~\cite{zhou2024pgmax} support only synchronous BP, restricting their applicability in program analysis.

Second, BP with local structures~\cite{wu2025belief} has not yet been explored in GPU-based implementations, although it is theoretically efficient. In program analysis, probabilistic constraints often exhibit inherent local structures, where many variable assignments share the same probability value. A typical example is the probabilistic Horn constraints in the form of \(p:x_1\land\cdots\land x_n\rightarrow y\), which is ubiquitous in modeling program semantics. While traditional BP requires an exponential number of operations (\(O(n\cdot2^n)\)) to marginalize over \(n\) variables in such a constraint, BP with local structures can reduce this complexity to linear time (\(O(n)\)) by grouping assignments with identical probabilities and extracting common terms. However, existing GPU-based approaches typically rely on dense tabular representations that fail to capture these computational redundancies.
Consequently, due to these limitations in generality and algorithmic efficiency, CPU-based approaches~\cite{wu2025belief,mooij2010libdai} still remain the preferred choice for program analysis applications.

We present \textsc{FastLBP}, a GPU-accelerated BP framework designed for program analysis. 
First, to accommodate diverse update strategies required in program analysis, we propose a unified representation for user-specified update strategies. We further develop a dependency analysis algorithm that identifies groups of messages that can be updated safely in parallel. By partitioning edges into ordered groups, \textsc{FastLBP} preserves the semantics of user-specified update strategies while enabling parallel execution. 
Second, to address the inefficiency caused by traditional BP, we incorporate BP with local structures into our GPU design. We organize messages and factor configurations in a compact layout to enable efficient memory access. We further design a GPU execution scheme that maps each message computation to a thread, where each thread evaluates a sequence of sub-messages. This design reduces the computational complexity by efficiently implementing BP with local structures.

We conduct extensive experiments to evaluate the efficiency and accuracy of our algorithm. We implement \textsc{FastLBP} in a prototype system and compare it with state-of-the-art approaches~\cite{wu2025belief,zhou2024pgmax}. We evaluate \textsc{FastLBP} using two representative approaches: SmartFL, a probabilistic fault localization approach, and Bingo, a Bayesian program analysis approach, to demonstrate the effectiveness of our approach. Experimental results show that: (1) compared to the state-of-the-art CPU-based approach~\cite{wu2025belief}, our approach achieves average speedups of $17.42\times$ and $2.82\times$ on SmartFL and Bingo, respectively, while maintaining the accuracy. Compared to the state-of-the-art GPU-based approach~\cite{zhou2024pgmax}, our approach achieves an average speedup of $6.14\times$ on SmartFL. (2) \textsc{FastLBP} supports update strategies used in Bingo that are not supported by existing GPU-based approaches, highlighting its advantage in enabling flexible update strategies.

In summary, our contributions are as follows:
\begin{itemize}
  \item We present \textsc{FastLBP}, a GPU-accelerated BP framework for program analysis, enabling flexible update strategies and efficient inference on large-scale PGMs.
  \item We propose a unified representation for user-specified update strategies and a dependency analysis algorithm for effective parallelization.
  \item We design a GPU implementation of BP with local structures by assigning each GPU thread to compute a single message and organizing the memory in a compact way.
  \item We implement a high-performance prototype and conduct extensive experiments to demonstrate the flexibility and efficiency of our approach.
\end{itemize}

\section{Motivating Example}
\label{sec:motivating-example}
In this section, we illustrate the key ideas of our approach using an example generated by Bingo~\cite{raghothaman2018user}, which performs BP on PGMs derived from analysis rules.

\begin{figure}[t]
\centering
\begin{tikzpicture}[
  node distance=0.5cm and 1cm,
  var/.style={circle, draw=red, minimum size=0.5cm, thick}, 
  factor/.style={rectangle, draw=red, minimum size=0.5cm, thick, fill=red!20}, 
  edge/.style={-, thick, red}
]
  
  \node[factor] (A5) {$a_5$};
  \node[var, below=of A5] (V5) {$v_5$}; 
  \node[var, above left=of A5] (V3) {$v_3$};
  \node[var, above right=of A5] (V4) {$v_4$}; 
  \node[factor, above=of V3] (A3) {$a_3$};
  \node[factor, above=of V4] (A4) {$a_4$};
  \node[var, above=of A4] (V2) {$v_2$};
  \node[var, above=of A3] (V1) {$v_1$};
  \node[factor, above=of V1] (A1) {$a_1$};
  \node[factor, above=of V2] (A2) {$a_2$};

  \draw[edge] (V1) -- (A1);
  \draw[edge] (V2) -- (A2);
  \draw[edge] (V1) -- (A3);
  \draw[edge] (V1) -- (A4);
  \draw[edge] (V2) -- (A4);
  \draw[edge] (V3) -- (A3);
  \draw[edge] (V4) -- (A4);
  \draw[edge] (V3) -- (A5);
  \draw[edge] (V4) -- (A5);
  \draw[edge] (V5) -- (A5);

% Legend

\node[factor, right=1.8cm of V5] (legend-factor) {};
\node[anchor=west, right=0.1cm of legend-factor] {\small Factor};

\node[var, above=0.4cm of legend-factor] (legend-var) {};
\node[anchor=west, right=0.1cm of legend-var] {\small Variable};
% \node[var, below=0.8cm of V5, xshift=-1.4cm] (legend-var) {};
% \node[anchor=west, right=0.1cm of legend-var] {\small Variable};

% \node[factor, below=0.8cm of V5, xshift=0.9cm] (legend-factor) {};
% \node[anchor=west, right=0.1cm of legend-factor] {\small Factor};

\end{tikzpicture}
\caption{Example of a factor graph.}
\Description{There are 5 factors and 5 variables in the factor graph.}
\label{fig:example-graph}
\end{figure}

\subsection{Support for Update Strategies}
\label{sec:update-strategy}
BP is a message-passing algorithm widely used for performing inference on factor graphs. 
\Cref{fig:example-graph} shows a subgraph of a factor graph generated by Bingo. A factor graph is a bipartite graph consisting of two types of nodes, variables and factors. For example, the variables are $V = \{v_1, v_2, v_3, v_4, v_5\}$ and the factors are $A = \{a_1, a_2, a_3, a_4, a_5\}$ in \Cref{fig:example-graph}. BP operates by iteratively passing messages along edges between variables and factors. A factor-to-variable message is computed depending on incoming variable-to-factor messages, while a variable-to-factor message is computed based on incoming factor-to-variable messages. For each message update, the incoming messages are those from neighboring nodes except the target node. For example, the message from $a_4$ to $v_1$, denoted as $\mu_{a_4\to v_1}$, depends on messages from $v_2$ and $v_4$ to $a_4$, denoted as $\mu_{v_2\to a_4}$ and $\mu_{v_4\to a_4}$. Moreover, $\mu_{v_2\to a_4}$ depends on $\mu_{a_2\to v_2}$, and $\mu_{v_4\to a_4}$ depends on $\mu_{a_5\to v_4}$. We treat variable-to-factor message computation as an intermediate step and therefore consider only dependencies among factor-to-variable messages. Before updating a factor-to-variable message, we first compute all of its required variable-to-factor messages using the corresponding incoming factor-to-variable messages. Consequently, in the example above, $\mu_{a_4\to v_1}$ depends on $\mu_{a_2\to v_2}$ and $\mu_{a_5\to v_4}$.

Update strategies play a critical role in the performance and convergence behavior of BP.
They define how message updates are scheduled. For instance, synchronous BP requires all messages to be updated in parallel, whereas asynchronous BP enforces a degree of sequentialism in message updates. Although synchronous BP is well suited for parallel execution, asynchronous BP has been shown to converge more reliably~\cite{elidan2006residual}, which is important for many program analysis applications. For example, Bingo ranks alarms according to the inferred probabilities. Asynchronous BP often produces more accurate probabilities than synchronous BP, thereby reducing the effort required for further inspection of bug reports. At the same time, synchronous BP remains useful because it runs faster and produces good results when it converges. 

However, existing approaches are unable to support flexible update strategies. One line of GPU-based work designs update strategies for specific domains, such as computer vision~\cite{grauergray2008gpu} and channel code decoding~\cite{romero2012sequential}. PGMax~\cite{zhou2024pgmax} is a Python package for running BP efficiently on both CPUs and GPUs and supports arbitrary graphs, but it only supports synchronous BP. Another line of work studies the convergence behavior of different update strategies in sequential and parallel scenarios~\cite{elidan2006residual,gonzalez2009residual,merwe2019message,aksenov2020relaxed}, but lacks a general mechanism to capture user-specified dependencies among message updates. Our work differs from these works by providing a systematic way to capture update dependencies, thereby enabling a wider range of update strategies for GPU parallel computation.

We illustrate a complex update strategy that is not supported by existing GPU-based approaches. To balance convergence stability and computational efficiency, Bingo adopts a round-robin asynchronous update strategy, where messages are updated sequentially in a predefined order during each iteration. When updating a message, the most recently updated neighboring messages are used. Although this sequential schedule improves convergence behavior, it introduces strict data dependencies between message updates.
These dependencies arise from both the user-specified order and the message update rules of BP. As a result, updating multiple messages simultaneously may violate the semantics of the asynchronous schedule.

Consider the factor graph in \Cref{fig:example-graph}. Suppose the factor-to-variable messages are updated in the following order during an iteration:
\[
\begin{aligned}
\mu_{a_1\to v_1},\;
\mu_{a_2\to v_2},\;
\mu_{a_3\to v_1},\;
\mu_{a_3\to v_3},\;
\mu_{a_4\to v_1},\\
\mu_{a_4\to v_2},\;
\mu_{a_4\to v_4},\;
\mu_{a_5\to v_3},\;
\mu_{a_5\to v_4},\;
\mu_{a_5\to v_5}.
\end{aligned}
\]
Let $\mu_{a\to v}^{i}$ denote the message from factor $a$ to variable $v$ at iteration $i$. 
In this example, $\mu_{a_4\to v_1}^{i + 1}$ depends on $\mu_{a_2\to v_2}^{i + 1}$ because the latter is updated earlier in the schedule, and on $\mu_{a_5\to v_4}^i$ because that message has not yet been updated. Such ordering constraints create dependencies that limit naive parallel execution. 
However, not all messages are mutually dependent. Two messages can be safely updated in parallel if neither requires the updated value of the other. In this example, $\mu_{a_2\to v_2}^{i+1}$ is independent of $\mu_{a_1\to v_1}^{i+1}$, and $\mu_{a_3\to v_1}^{i+1}$ is independent of both. Therefore, these messages can be updated concurrently without violating the asynchronous semantics. However, existing GPU-based approaches fail to capture such dependencies.

To enable parallel execution while preserving the semantics of different user-specified update strategies, our approach first provides a unified representation of update strategies, namely a preordered set $(E, \lesssim)$ over the edge set $E$ of the factor graph. Intuitively, specifying $e_1\lesssim e_2$ means that, in a single iteration, the update of the message on edge $e_1$ must happen no later than that on edge $e_2$. We then perform dependency analysis that captures both the update strategy constraints and BP computation dependencies. Based on this analysis, messages can be partitioned into ordered groups that can be updated in parallel. 

In this example, our approach represents the update strategy as a total order induced by the preorder according to the updating sequence:
\[
\begin{aligned}
&(a_1, v_1)\prec
(a_2, v_2)\prec
(a_3, v_1)\prec
(a_3, v_3)\prec
(a_4, v_1)\prec\\
&(a_4, v_2)\prec
(a_4, v_4)\prec
(a_5, v_3)\prec
(a_5, v_4)\prec
(a_5, v_5).
\end{aligned}
\]
Our approach then performs a topological traversal on the Hasse diagram of this order and partitions the messages into groups. 
Specifically, for $(a, v)\prec (a', v')$, our approach ensures that if $\mu_{a'\to v'}$ depends on $\mu_{a\to v}$ in the BP computation, then $\mu_{a\to v}$ must be computed strictly earlier than $\mu_{a'\to v'}$; otherwise, the two updates can be performed in parallel.
% Otherwise, the updates can be performed in parallel if either $\mu_{a\to v}$ and $\mu_{a'\to v'}$ are independent, or $(a, v)\prec (a', v')$ but $\mu_{a'\to v'}$ does not depend on $\mu_{a\to v}$.
First, we can update $\mu_{a_1\to v_1}$, $\mu_{a_2\to v_2}$, and $\mu_{a_3\to v_1}$ in parallel, as we discussed earlier in this section. Next, we place $\mu_{a_3\to v_3}$ into a new group because it depends on $\mu_{a_1\to v_1}$. We then add $\mu_{a_4\to v_1}$, $\mu_{a_4\to v_2}$, and $\mu_{a_4\to v_4}$ to the same group one by one, because these four messages are mutually independent. But we place $\mu_{a_5\to v_3}$ into another group because it depends on $\mu_{a_4\to v_4}$. Finally, $\mu_{a_5\to v_4}$ and $\mu_{a_5\to v_5}$ can be placed in the same group as $\mu_{a_5\to v_3}$.
As a result, the messages can be partitioned into three groups:
\[
\begin{gathered}
\{\mu_{a_1\to v_1},\; \mu_{a_2\to v_2},\; \mu_{a_3\to v_1}\}, \\
\{\mu_{a_3\to v_3},\; \mu_{a_4\to v_1},\; \mu_{a_4\to v_2},\; \mu_{a_4\to v_4}\}, \\
\{\mu_{a_5\to v_3},\; \mu_{a_5\to v_4},\; \mu_{a_5\to v_5}\}.
\end{gathered}
\]
Updating each group in parallel, while processing the groups sequentially, preserves the semantics of the original asynchronous schedule and exposes substantial parallelism.

\subsection{Implementing BP with Local Structures}
\label{sec:local-structure}
Many probabilistic constraints arising in program analysis exhibit local structures~\cite{wu2025belief}. 
Local structures occur when multiple assignments of variables share the same probability value in a factor. 
Instead of representing the factor as a full table over all assignments, these repeated patterns can be grouped and computed more efficiently.
A common example in program analysis is the probabilistic Horn constraint
$p : x_1 \land x_2 \land \cdots \land x_n \rightarrow y$. 
Such a constraint states that if all premises hold, the conclusion holds with probability $p$, and 0 otherwise. 
% When the constraint is encoded as a factor in a factor graph, many assignments of the variables share identical probabilities. 
When the constraint is encoded as a factor in a factor graph, the resulting factor is represented by a factor function $f_a$, which maps each joint assignment $(x_1, x_2, \ldots,x_n,y)$ to a probability value. Many such assignments of the variables share identical probabilities, constituting the local structure in $f_a$.
Given $f_a$, traditional BP computes a factor-to-variable message by enumerating all assignments of neighboring variables, resulting in an exponential complexity of $O(n\cdot2^n)$. 
BP with local structures exploits the repeated patterns in the factor to group assignments with identical probabilities, reducing the computation to linear time $O(n)$.

To leverage this advantage on GPUs, we extend the message-level parallelization of BP to support BP with local structures. We assign each GPU thread to compute a single factor-to-variable message. 
% By providing each thread with the corresponding edge, factor, and variable identifiers, we can efficiently access the required messages and factor configurations.
% \Cref{fig:computation-process} illustrates the computation process using the message group:
% \[
% \{\mu_{a_3\to v_3},\; \mu_{a_4\to v_1},\; \mu_{a_4\to v_2},\; \mu_{a_4\to v_4}\}\text{.}
% \]
In BP with local structures, a factor-to-variable message is computed as the sum of several sub-messages. 
Each sub-message corresponds to a group of assignments that share the same probability value.
It is either zero or computed as a constant multiplied by the product of incoming variable-to-factor messages.
% For instance, for the constraints generated by Bingo, messages are defined over binary variables (domain $\{0, 1\}$). Each entry is computed as the sum of five sub-messages, resulting in ten sub-messages in total.
For instance, for the constraints generated by Bingo, the variables are binary (domain $\{0, 1\}$). Each message has two values, one for each possible assignment of the corresponding variable. According to the local structures in Bingo, each value is computed as the sum of five sub-messages, requiring ten sub-message computations for the entire message.
\Cref{fig:computation-process} illustrates the computation process for the second message group from the running example in \Cref{sec:update-strategy}:
\[
\{\mu_{a_3\to v_3},\; \mu_{a_4\to v_1},\; \mu_{a_4\to v_2},\; \mu_{a_4\to v_4}\}\text{.}
\]

\begin{figure}[t]
\centering
\begin{tikzpicture}[
  font=\small,
  cell/.style={draw, minimum width=0.65cm, minimum height=0.6cm},
  idle/.style={cell, fill=gray!30},
  exec/.style={cell, fill=yellow!50},
  label/.style={anchor=west}
]

\draw[
    decorate,
    line width=1.2pt,
    decoration={snake, amplitude=2pt, segment length=6pt}
]
(-0.8,-1.6) -- (-0.8,0.8);
\draw[->, thick] (-0.325,1.0) -- (6.175,1.0) node[midway, above] {time};

\node[label] at (6.2,0.5) {$\mu_{a_3 \rightarrow v_3}$};
\node[label] at (6.2,-0.1) {$\mu_{a_4 \rightarrow v_1}$};
\node[label] at (6.2,-0.7) {$\mu_{a_4 \rightarrow v_2}$};
\node[label] at (6.2,-1.3) {$\mu_{a_4 \rightarrow v_4}$};

% row 1
\foreach \i in {0,...,9}{
    \node[exec] at (\i*0.65,0.5) {$\text{sm}_\i$};
}

% row 2
\foreach \i in {0,...,9}{
    \node[exec] at (\i*0.65,-0.1) {$\text{sm}_\i$};
}

% row 3
\foreach \i in {0,...,9}{
    \node[exec] at (\i*0.65,-0.7) {$\text{sm}_\i$};
}

\foreach \i in {0,...,9}{
    \node[exec] at (\i*0.65,-1.3) {$\text{sm}_\i$};
}

\end{tikzpicture}

\caption{The computation process of our approach. Each row represents one GPU thread
computing the message shown on the right.
Within each row, $\mathrm{sm}_0,\ldots,\mathrm{sm}_9$ denote sub-message
computations executed sequentially from left to right.}
\Description{A slot indicates a sub-message computation}
\label{fig:computation-process}
\end{figure}

\begin{figure*}[t]
\centering
\begin{tikzpicture}[
  font=\small,
  cell/.style={draw, minimum width=1.2cm, minimum height=0.6cm},
  exec/.style={cell, fill=orange!30},
  highlight/.style={draw=red, thick, dashed},
  label/.style={anchor=west}
]

\node[exec] at (0*1.2,1.7) {$t_{(a_1, v_1)}$};
\node[exec] at (1*1.2,1.7) {$t_{(a_2, v_2)}$};
\node[exec] at (2*1.2,1.7) {$t_{(a_3, v_1)}$};
\node[exec] at (3*1.2,1.7) {$t_{(a_3, v_3)}$};
\node[exec] (tv1a4) at (4*1.2,1.7) {$t_{(a_4, v_1)}$};
\node[exec] (tv2a4) at (5*1.2,1.7) {$t_{(a_4, v_2)}$};
\node[exec] (tv4a4) at (6*1.2,1.7) {$t_{(a_4, v_4)}$};
\node[exec] at (7*1.2,1.7) {$t_{(a_5, v_3)}$};
\node[exec] at (8*1.2,1.7) {$t_{(a_5, v_4)}$};
\node[exec] at (9*1.2,1.7) {$t_{(a_5, v_5)}$};

% row 1
\node[exec] at (0*1.2,0.5) {$\mu_{a_1\to v_1}$};
\node[exec] at (1*1.2,0.5) {$\mu_{a_3\to v_1}$};
\node[exec] (a4v1) at (2*1.2,0.5) {$\mu_{a_4\to v_1}$};
\node[exec] at (3*1.2,0.5) {$\mu_{a_2\to v_2}$};
\node[exec] at (4*1.2,0.5) {$\mu_{a_4\to v_2}$};
\node[exec] at (5*1.2,0.5) {$\mu_{a_3\to v_3}$};
\node[exec] at (6*1.2,0.5) {$\mu_{a_5\to v_3}$};
\node[exec] at (7*1.2,0.5) {$\mu_{a_4\to v_4}$};
\node[exec] at (8*1.2,0.5) {$\mu_{a_5\to v_4}$};
\node[exec] at (9*1.2,0.5) {$\mu_{a_5\to v_5}$};

% row 2
\node[exec] at (0*1.2,-0.7) {$\mu_{v_1\to a_1}$};
\node[exec] at (1*1.2,-0.7) {$\mu_{v_2\to a_2}$};
\node[exec] at (2*1.2,-0.7) {$\mu_{v_1\to a_3}$};
\node[exec] at (3*1.2,-0.7) {$\mu_{v_3\to a_3}$};
\node[exec] (v1a4) at (4*1.2,-0.7) {$\mu_{v_1\to a_4}$};
\node[exec] (v2a4) at (5*1.2,-0.7) {$\mu_{v_2\to a_4}$};
\node[exec] (v4a4) at (6*1.2,-0.7) {$\mu_{v_4\to a_4}$};
\node[exec] at (7*1.2,-0.7) {$\mu_{v_3\to a_5}$};
\node[exec] at (8*1.2,-0.7) {$\mu_{v_4\to a_5}$};
\node[exec] at (9*1.2,-0.7) {$\mu_{v_5\to a_5}$};

\node[highlight, fit=(v2a4)(v4a4)] (agg) {};
\node[highlight, fit=(tv1a4)(tv2a4)(tv4a4)] (tagg) {};
\node[highlight, fit=(a4v1)] (a4v1h) {};

\draw[->, thick, red] (agg) -- (a4v1h);
\draw[->, thick, red] (tagg) -- (a4v1h);

\end{tikzpicture}

\caption{Example of the memory layout and access. Factor configurations for factor $a$ and variable $v$ are denoted as $t_{(a, v)}$.}
\Description{Three horizontal arrays show factor configurations,
factor-to-variable messages, and variable-to-factor messages. Red dashed
boxes highlight the configurations for factor $a_4$, the target message
$\mu_{a_4\to v_1}$, and its incoming messages $\mu_{v_2\to a_4}$ and
$\mu_{v_4\to a_4}$. Arrows indicate the data accessed to compute the target
message.}
\label{fig:memory-access}
\end{figure*}

Another problem is how to represent factor graphs efficiently, as those generated by program analysis applications are typically large and sparse. 
To achieve memory-efficient storage and efficient access, we adopt the Compressed Sparse Row (CSR)~\cite{saad2003iterative} format to represent messages. 
In CSR, a sparse matrix is stored as a flat row-major array.
To compute a factor-to-variable message $\mu_{a\to v}$, we require all incoming variable-to-factor messages to factor $a$. 
Therefore, the variable-to-factor message matrix is organized such that rows are indexed by factor identifiers, grouping messages associated with the same factor contiguously. 
Conversely, to compute a variable-to-factor message $\mu_{v\to a}$, we require all incoming factor-to-variable messages to variable $v$, and thus the factor-to-variable message matrix is indexed by variable identifiers.
For example, the first row (corresponding to $v_1$) contains  $\{\mu_{a_1\to v_1}, \mu_{a_3\to v_1}, \mu_{a_4\to v_1}\}$, and the second row (corresponding to $v_2$) contains $\{\mu_{a_2\to v_2}, \mu_{a_4\to v_2}\}$.
Additionally, factor configurations are stored in a similar manner to variable-to-factor messages. 
Unlike standard BP, which enumerates over joint assignments of variables, BP with local structures computes sub-messages based on an indicator specifying whether a variable is fixed. 
Instead of using a dense tabular representation, we store these configurations in a flat array.
By providing each thread with the corresponding edge, factor, and variable identifiers, we can efficiently access the required messages and factor configurations.
\Cref{fig:memory-access} illustrates the memory layout of the factor graph in \Cref{fig:example-graph} and how $\mu_{a_4\to v_1}$ accesses the required messages and factor configurations.

\section{Preliminaries}
\label{sec:preliminaries}

In this section, we introduce the background of our approach, including formal definitions of factor graphs and belief propagation, as well as a brief overview of GPU programming.

\subsection{Factor Graph}
\label{sec:factor-graph}

A factor graph is a special type of PGM. Common PGMs such as Bayesian networks and Markov random fields can be converted into factor graphs, enabling a unified approach to inference.
A factor graph is a bipartite graph representing the factorization of a joint probability distribution. Given a set of random variables $\mathbf{X} = \{X_1, X_2, \dots, X_n\}$, suppose their joint probability distribution can be factorized into a product of functions:
\[
P(\mathbf{X}) = \frac{1}{\mathcal{Z}} \prod_{a \in A} f_a(\mathbf{X}_a),
\]
where $\mathbf{X}_a \subseteq \mathbf{X}$ is a subset of $\mathbf{X}$, and $\mathcal{Z}$ is the normalization constant. The corresponding factor graph $G = (V, A, E)$ consists of variable nodes $V = \{v_1, v_2, \dots, v_n\}$, factor nodes $A = \{a_1, a_2, \dots, a_m\}$, and edges $E \subseteq V \times A$. Each node $v_i$ corresponds to the random variable $X_i$ (or $X_{v_i}$ for notational convenience) and each node $a_j$ corresponds to the function $f_{a_j}$. An edge $(v_i, a_j) \in E$ exists if and only if $X_i \in \mathbf{X}_{a_j}$.

\subsection{Belief Propagation}
\label{sec:belief-propagation}

Belief Propagation (BP) is a message-passing algorithm for performing inference on factor graphs. 
BP iteratively exchanges messages between variables and factors along the edges of a factor graph and computes marginal distributions of variables with these messages.
For each edge $(v, a) \in E$, two types of messages are defined, a message $\mu_{v\to a}$ from $v$ to $a$, and a message $\mu_{a \to v}$ from $a$ to $v$. Both types of messages, $\mu_{v\to a}, \mu_{a \to v}: \text{Range}(X_v)\to\mathbb{R}$, are real-valued functions whose domain is the set of values that can be taken by the random variable associated with $v$.

The message sent from a variable $v$ to a neighboring factor $a$ is defined as the product of the incoming factor-to-variable messages:
\begin{equation}
\label{eq:message-vtoa}
    \mu_{v \to a}(x_v) = \prod_{a^* \in N_V(v) \setminus \{a\}} \mu_{a^* \to v}(x_v),
\end{equation}
where $N_V(v)$ denotes the set of neighboring factors of $v$.

The message sent from a factor $a$ to a neighboring variable $v$ is defined as the product of the factor function with the incoming variable-to-factor messages:
\begin{equation}
\label{eq:message-atov}
    \mu_{a \to v}(x_v) = \sum_{\mathbf{x}_a \setminus x_v} f_a(\mathbf{x}_a) \prod_{v^* \in N_A(a) \setminus \{v\}} \mu_{v^* \to a}(x_{v^*}),
\end{equation}
where $N_A(a)$ denotes the set of neighboring variables of $a$, and $\sum_{\mathbf{x}_a \setminus x_v}$ denotes summation over all possible values of $\mathbf{x}_a$ while keeping $x_v$ fixed.

After convergence or a certain number of iterations, the marginal probability of a variable $v$ is computed as
\begin{equation}
\label{eq:marginal}
    P(X_v = x_v) = \frac{\prod_{a \in N_V(v)} \mu_{a \to v}(x_v)}{\sum_{x_v} \prod_{a \in N_V(v)} \mu_{a \to v}(x_v)}.
\end{equation}

If the factor graph is a tree, BP yields exact marginal distributions after a finite number of iterations. 
For graphs with cycles, BP is not guaranteed to converge but often provides good approximate solutions in practice.

\subsection{BP with Local Structures}
\label{sec:bp-local-structure}
The main computational bottleneck of BP lies in the factor-to-variable message $\mu_{a\to v}$ in \Cref{eq:message-atov}, which requires enumeration over an exponential number of factor values.
Local structures in a constraint can be exploited to simplify this computation by grouping assignments that share the same factor value. The computation of $\mu_{a\to v}$ is then reduced to the summation of multiple sub-messages based on the local structure of the factor function. Formally,
\begin{equation}
\label{eq:wu-atov}
\mu_{a \rightarrow v}(x_v) = \sum_{k=1}^{l}
\sum_{\mathbf{x}_a \backslash x_v}
t_k(\mathbf{x}_a) \prod_{v^* \in N_A(a) \setminus \{v\}} \mu_{v^* \to a}(x_{v^*})\text{,}
\end{equation}
% where each $t_k$ specifies a weight together with indicators indicating whether a variable is constrained to a fixed value or treated as a wildcard. 
where each $t_k$ specifies a weight together with indicators that determine whether a variable is fixed to a particular value or left unrestricted. If a variable is unrestricted, its incoming message is summed over all possible values, as shown in lines~\ref{line:unrestricted-start}-\ref{line:unrestricted-end} of \Cref{alg:one_m}.
According to \Cref{eq:wu-atov}, each message is computed using \Cref{alg:one_m}, where $v$ is the target variable index, $x_v$ is the target value, $p[l]$ denotes the weights associated with the factor, $b = |N_A(a)|$ is the number of variables adjacent to $a$, and $l$ is the number of sub-messages. The indicator arrays are $indicator[b][l]$, and $\mu[b][r]$ stores messages from $N_A(a)$ to $a$.
As shown in \Cref{alg:one_m}, the complexity is reduced from $O(b \cdot r^b)$ to $O(l \cdot b \cdot r)$, where $b$ can be large.
\begin{algorithm}[t]
\caption{Message Calculation}
\label{alg:one_m}
\begin{algorithmic}[1]
\Require $v$, $x_v$, $p[l]$, $indicator[b][l]$, $\mu[b][r]$
\Ensure $message$

\State $message\gets 0$
\For{$t\gets 1$ to $l$}
\If{$indicator[v][t] \neq x_v \land indicator[v][t] \neq -1$}
    \State \textbf{continue}
\Else

\State $sm\gets p[t]$
\For{$j\gets 1$ to $b$, $j\neq v$}
        \State $c\gets indicator[j][t]$
        \If{$c = -1$} \label{line:unrestricted-start}
            \State $sm\gets sm \times \sum_{k=1}^r \mu[j][k]$ \label{line:unrestricted-end}
        \Else
            \State $sm\gets sm \times \mu[j][c]$
        \EndIf
\EndFor
\State $message\gets message + sm$

\EndIf
\EndFor
\State \Return $message$
\end{algorithmic}
\end{algorithm}

\subsection{GPU Programming}
\label{sec:gpu-programming}

We introduce GPU programming using Compute Unified Device Architecture (CUDA) as a representative example. Modern GPUs consist of multiple streaming multiprocessors (SMs), each of which contains many compute units for massive parallel execution. 
The functions executed on GPUs are referred to as kernels. A kernel is launched to run multiple threads in parallel. These threads are grouped into thread blocks, while these blocks together form a grid. Each thread block is scheduled to run on a single SM. Within a thread block, threads are divided into groups of 32 threads known as warps. Threads within a warp execute instructions in a Single-Instruction Multiple-Threads (SIMT) paradigm. 

\section{\textsc{FastLBP} Framework}
\label{sec:framework}

In this section, we first describe the overall workflow of \textsc{FastLBP} (\Cref{sec:workflow}). 
We then present two key designs of \textsc{FastLBP}: (1) the update strategy representation with the dependency analysis algorithm (\Cref{sec:representation}), and (2) the implementation of BP with local structures on GPUs (\Cref{sec:scheduling}).

\subsection{Overall Workflow}
\label{sec:workflow}

\Cref{fig:workflow} shows the heterogeneous workflow of \textsc{FastLBP}. Given a factor graph derived from probabilistic program analysis and a user-defined update strategy, \textsc{FastLBP} ultimately produces marginal probabilities of variables as results. The workflow consists of the following four major components.

\paragraph{Initialization.} 
\textsc{FastLBP} first takes as input a factor graph, typically generated by probabilistic program analysis frameworks such as Bingo or SmartFL. Based on the graph structure and factor functions, \textsc{FastLBP} allocates GPU memory for messages using a CSR-based layout. This representation enables compact storage and efficient access for sparse message structures, avoiding the prohibitive memory cost of dense adjacency matrices on GPUs.

\paragraph{Dependency Analysis.} 
\textsc{FastLBP} requires the user to specify an update strategy. \textsc{FastLBP} represents the update strategy as a preordered set $(E, \lesssim)$ over graph edges. This abstraction uniformly captures different update strategies. It then analyzes the dependencies implied by this preordered set and the BP computation process. The result is a partitioning of edges into ordered groups, each representing a set of factor-to-variable messages that can be updated simultaneously.

\paragraph{Kernel Launch.} 
Based on the results of dependency analysis, \textsc{FastLBP} organizes message updates into batches corresponding to the computed edge groups. For each group, a GPU kernel is launched, where each thread is assigned to compute a single message. The required factor, variable, and edge identifiers are passed to each thread to enable efficient data access. This batching strategy reduces kernel launch overhead and ensures that each kernel exposes sufficient parallelism for efficient GPU execution.

\paragraph{Message Passing.} 
Finally, \textsc{FastLBP} performs iterative message passing on the GPU. In each iteration, kernels are launched following the dependency order, and messages are updated accordingly using the message-level parallel implementation. The process continues until convergence or a user-specified number of iterations, after which marginal probabilities are computed and returned to the analysis framework.

\begin{figure}[t]
\centering
\begin{tikzpicture}[
    node distance=8mm,
    >=latex,
    box/.style={
        rectangle,
        rounded corners,
        draw=black,
        thick,
        minimum width=3.5cm,
        minimum height=1cm,
        align=center
    },
    cpu/.style={
        box,
        fill=blue!8
    },
    gpu/.style={
        box,
        fill=red!12
    }
]

% ----- Main pipeline -----
\node[cpu] (init) {Initialization};

\node[cpu, below=of init] (dep) {Dependency Analysis};

\node[cpu, below=of dep] (grouping) {Kernel Launch};

\node[gpu, below=of grouping] (mp) {Message Passing};

% ----- Main execution flow (solid arrows) -----
\draw[->] (init) -- (dep);
\draw[->] (dep) -- (grouping);
\draw[->] (grouping) -- (mp);

% ----- Labeled dashed input arrows (pointing RIGHT) -----
\draw[->, dashed]
($(init.west)+(-2,0)$) -- (init.west)
node[midway, above] {Factor}
node[midway, below] {Graph};

\draw[->, dashed]
($(dep.west)+(-2,0)$) -- (dep.west)
node[midway, above] {Update}
node[midway, below] {Strategy};

% ----- Labeled dashed output arrow (pointing LEFT) -----
\draw[->, dashed]
(mp.west) -- ++(-2,0)
node[midway, above] {Results};

% ----- GPU memory allocation (curved dashed arrow) -----
\draw[->, dashed]
(init.east) .. controls +(1,0) and +(1,0)
.. node[right, align=center]
{GPU\\Memory\\Allocation}
(mp.east);

\node[
    rectangle,
    rounded corners,
    draw=black,
    thick,
    dashed,
    inner sep=6pt,
    fit=(init)(dep)(grouping)
] (cpubox) {};

\node[
    fill=white,
    inner sep=2pt
] at (cpubox.north) {CPU};

% ----- GPU group box -----
\node[
    rectangle,
    rounded corners,
    draw=black,
    thick,
    dashed,
    inner sep=6pt,
    fit=(mp)
] (gpubox) {};

\node[
    fill=white,
    inner sep=2pt
] at (gpubox.south) {GPU};

\end{tikzpicture}
\caption{Workflow of \textsc{FastLBP}.}
\label{fig:workflow}
\Description{Workflow of \textsc{FastLBP}. The CPU performs initialization,
dependency analysis, and kernel launch, while the GPU performs message
passing. The factor graph and update strategy are inputs, GPU memory is
allocated for message passing, and the computed marginal probabilities are
returned as results.}
\end{figure}

\subsection{Representation and Dependency Analysis}
\label{sec:representation}

In this section, we first briefly introduce the representation of update strategies, then formalize its definition and semantics, and finally present the dependency analysis algorithm.

\subsubsection{Update Strategies} 
An update strategy in \textsc{FastLBP} is represented as a preordered set $(E, \lesssim)$ over graph edges.
\Cref{tab:update-strategy} describes several typical update strategies. Users can specify the strategy name, and \textsc{FastLBP} encodes it together with the edge set into a unified intermediate representation, i.e., a preordered set. Additionally, users may define a custom preordered set directly as the update strategy.
In practice, this is typically a small adapter that maps message edges to their order in the desired schedule.

\begin{table}[t]
\caption{Update Strategies}
\label{tab:update-strategy}
\centering
\begin{tabular}{@{}ll@{}}
\toprule[1.5pt]
Strategy Name & Description \\
\midrule[1pt]
PARALL & Update all messages in parallel \\
SEQFIX & Sequentially update messages in a fixed order \\
TOPO   & Update messages according to topological order \\
$\ldots$ & $\ldots$ \\
\bottomrule[1.5pt]
\end{tabular}
\end{table}

\subsubsection{Representation Formalization} 
% In this section, we first give the formal definition of the intermediate representation of update strategies. Then we build the relationship between the representation and update dependencies. We end this section with the theorem of the legality of the representation.

To perform dependency analysis in a uniform manner, we introduce a unified representation of update dependencies. Our key insight is that update strategies can be uniformly characterized as preorders over the edge set of a factor graph. This abstraction allows different scheduling semantics to be handled within a single formal framework. 

We first characterize the dependencies induced by different update strategies. The problem lies in whether a message depends on messages from the current iteration or the previous iteration. 
We focus on factor-to-variable updates, since they are directly used to compute marginal probabilities. We assume that each factor-to-variable message is updated once per iteration. When computing a factor-to-variable message, all incoming variable-to-factor messages must first be computed. These variable-to-factor messages, in turn, depend on the incoming factor-to-variable messages. By treating the computation of variable-to-factor messages as an intermediate step, the essential dependency lies in whether the required factor-to-variable messages come from the current or the previous iteration. 
% To make this explicit, we introduce a function $\delta: E \times E \to \{0,1\}$, where $\delta(e_1, e_2)=1$ indicates that the update on $e_1$ depends on the message on $e_2$ from the current iteration, and $\delta(e_1, e_2)=0$ otherwise.
Formally, by substituting Equation~(\ref{eq:message-vtoa}) into Equation~(\ref{eq:message-atov}) and introducing superscripts to indicate iteration indices, we obtain:
\begin{equation}
\label{eq:message-dependence}
\begin{aligned}
&\mu_{a\to v}^{i+1}(x_v) = \\
&\sum_{\mathbf{x}_a\setminus x_v}
f_a(\mathbf{x}_a)
\prod_{v^* \in N_A(a)\setminus \{v\}}
\prod_{a^* \in N_V(v^*)\setminus \{a\}}
\mu_{a^*\to v^*}^{i+\delta((a, v), (a^*, v^*))}(x_{v^*})\text{,}
\end{aligned}
\end{equation}
where $\delta: E \times E \to \{0, 1\}$ indicates whether a dependency comes from the current iteration ($\delta=1$) or the previous iteration ($\delta=0$). 

This expression shows that the update of a message on edge $(a,v)$ depends only on a specific subset of edges determined by the graph structure. To make this explicit, we define the edge dependency neighborhood. 
\begin{definition}[Edge Dependency Neighborhood]
\label{def:edge-neighborhood}
Define the edge dependency neighborhood function $N_E: E \to \mathcal{P}(E)$ as
\begin{equation*}
    N_E((a,v)) \overset{\text{def}}{=} \bigcup_{v^* \in N_A(a) \setminus \{v\}} \bigcup_{a^* \in N_V(v^*) \setminus \{a\}} \{(a^*, v^*)\}.
\end{equation*}
\end{definition}
% A valid update strategy must specify a legitimate $\delta$ function such that the computation remains well-defined and respects the intended scheduling semantics. 
The key to defining a legitimate update strategy is specifying a reasonable $\delta$. We now give the definition of intermediate representation used in \textsc{FastLBP} for update strategies, along with the corresponding definition of $\delta$. 

\begin{definition}[Intermediate Representation of Update Strategies]
\label{def:update-strategy}
The intermediate representation of an update strategy is defined as a preordered set $(E, \lesssim)$ on the edge set $E$ of a factor graph, such that for all $e_1\in E$ and for all $e_2 \in N_E(e_1)$, either $e_1 \lesssim e_2$ or $e_2 \lesssim e_1$ holds.
Then, given $(E, \lesssim)$, for all $e_1, e_2 \in E$,
\begin{equation*}
    \delta(e_1, e_2) \overset{\text{def}}{=} 
    \begin{cases}
        1, & \text{if } e_2 \prec e_1 \text{ and } e_2 \in N_E(e_1), \\
        0, & \text{otherwise},
    \end{cases}
\end{equation*}
where $\prec$ is the strict partial order induced by $\lesssim$. The message update is performed according to Equation~(\ref{eq:message-dependence}).
\end{definition}
Intuitively, specifying that $e_2 \lesssim e_1$ in the preorder means that, during one iteration, the message on edge $e_1$ cannot be updated before the message on edge $e_2$. Specifically, $e_2 \prec e_1$ means the message on $e_2$ must be updated before the message on $e_1$ if $e_1$ depends on $e_2$ ($\delta(e_1, e_2) = 1$). 

The following Theorem~\ref{thm:strategy-legality} shows that the function $\delta$ defined in Definition~\ref{def:update-strategy} is legitimate.

\begin{theorem}[Legality of Update Strategy Representation]
\label{thm:strategy-legality}
If the function $\delta$ is given by Definition~\ref{def:update-strategy}, then there does not exist a sequence of edges $e_{j_1}, e_{j_2}, \ldots, e_{j_p} \in E$ (with $p > 1$) such that
\begin{align*}
    e_{j_1} = e_{j_p} \quad \text{and} \quad \forall\, 1 \leq i < p,\ \delta(e_{j_i}, e_{j_{i+1}}) = 1.
\end{align*}
\end{theorem}
\begin{proof}
We proceed by contradiction. Suppose such a sequence exists.

When $p = 2$, we have $e_{j_1} = e_{j_2}$ and $\delta(e_{j_1}, e_{j_2}) = 1$. However, by Definition~\ref{def:update-strategy}, $\delta(e_{j_1}, e_{j_2}) = 1$ implies $e_{j_2} \prec e_{j_1}$. Since $e_{j_1} = e_{j_2}$, this implies $e_{j_1} \prec e_{j_1}$, which contradicts the irreflexivity of the strict partial order $\prec$.

For $p > 2$, by Definition~\ref{def:update-strategy}, $\delta(e_{j_i}, e_{j_{i+1}}) = 1$ for all $1 \leq i < p$ implies $e_{j_i} \prec e_{j_{i+1}}$. Thus, by transitivity of $\prec$, it follows that $e_{j_1} \prec e_{j_p}$. But since $e_{j_1} = e_{j_p}$, this yields $e_{j_1} \prec e_{j_1}$, again contradicting the irreflexivity of $\prec$.

Therefore, no such cyclic sequence can exist, completing the proof.
\end{proof}

Theorem~\ref{thm:strategy-legality} shows that the update strategy defined by Definition~\ref{def:update-strategy} contains no cyclic dependencies. That is, within a single iteration, starting from any message, although it may depend on other messages computed in the current iteration, and those messages in turn may depend on yet others from the same iteration, this chain of dependencies must eventually terminate at messages whose values are taken from the previous iteration. This termination is guaranteed by the finiteness of the edge set $E$. Therefore, the current iteration's results can always be computed starting from the messages of the previous iteration, which confirms the well-definedness of Definition~\ref{def:update-strategy}. 

Definition~\ref{def:update-strategy} captures a sufficiently broad subset of valid update strategies. For example, PARALL can be represented by a preorder in which all edges are in the same equivalence class, SEQFIX corresponds to a total order over all edges, and TOPO is expressed as a topological order.

\subsubsection{Dependency Analysis Algorithm.}
We now present the dependency analysis algorithm based on the intermediate representation of the update strategy.

\begin{algorithm}[t]
    \caption{Dependency Analysis}
    \label{alg:dependence-analysis}
    \begin{algorithmic}[1]
        \Require A preordered set $(E, \lesssim)$ represented as a Hasse diagram, where $E = \{e_1,\ldots,e_N\}$
        \Ensure A sequence of edge sets $S = [s_1,\ldots,s_k]$
        \State $e_1,\ldots,e_N \gets \operatorname{TopologicalSort}(E, \lesssim)$ \label{line:dependence-toposort}
        \State $S \gets \emptyset$, $NextSet \gets \emptyset$
        \For{$i \gets 1$ to $N$} \label{line:dependence-outerloop}
            \For{$e$ in $NextSet$}
                \If{$e \in N_E(e_i)$ and $e \prec e_i$} \label{line:dependence-if}
                    \State $S.\text{append}(NextSet)$
                    \State $NextSet \gets \emptyset$
                    \State \textbf{break}
                \EndIf \label{line:dependence-endif}
            \EndFor
            \State $NextSet \gets NextSet \cup \{e_i\}$
        \EndFor
        \State $S.\text{append}(NextSet)$
        \State \Return S
    \end{algorithmic}
\end{algorithm}

In \Cref{alg:dependence-analysis}, $\operatorname{TopologicalSort}$ on line~\ref{line:dependence-toposort} refers to a topological sort performed on the Hasse diagram of the partial order on the set of equivalence classes. Edges in the same equivalence class can be in any order. After the topological sort, the algorithm greedily groups the edges into batches. When processing the edge $e_i$, if it depends on any edge already in $NextSet$, then the current $NextSet$ is output as a batch, and a new batch is started. Finally, the algorithm outputs a sequence of edge sets $S$. Updating factor-to-variable messages according to this sequence yields the expected result under the specified update strategy. 
The overall time complexity of the dependency analysis is $O(|E|^2)$ in the worst case, as the algorithm performs a quadratic scan in the number of edges. Since this analysis is executed only once prior to message iterations, its cost is negligible compared to the overall inference procedure.

Theorem~\ref{thm:correctness} establishes the correctness of the algorithm, guaranteeing that all message dependencies for a given batch are correctly sourced, either from previously computed batches in the current iteration or from the previous iteration, as required.
\begin{theorem}[Correctness]
\label{thm:correctness}
The dependency analysis algorithm is correct in the following sense: for all $s_j \in \{s_1, \ldots, s_k\}$, for all $e \in s_j$, and for all $e' \in N_E(e)$,
\begin{enumerate}
    \item if $e' \prec e$, then $e' \in \bigcup_{1 \le l < j} s_l$.
    \item if $e \lesssim e'$, then $e' \in \bigcup_{j \le l \le k} s_l$.
\end{enumerate}
\end{theorem}
\begin{proof}
We prove the theorem in two parts:
\begin{itemize}
    \item Since the batches are constructed according to a topological order (line~\ref{line:dependence-toposort}), if $e' \prec e$, then $e' \in \bigcup_{1 \le l \le j} s_l$. Moreover, because $e' \in N_E(e)$ and $e' \prec e$, lines~\ref{line:dependence-if}--\ref{line:dependence-endif} ensure that $e' \notin s_j$. Combining both facts, we conclude that $e' \in \bigcup_{1 \le l < j} s_l$.
    \item If $e \lesssim e'$, by the topological sort (line~\ref{line:dependence-toposort}), $e'$ cannot appear in any batch preceding $s_j$. Therefore, $e' \in \bigcup_{j \le l \le k} s_l$.
\end{itemize}
\end{proof}

\subsection{Implementation of BP with Local Structures}
\label{sec:scheduling}

To leverage the computational benefits of BP with local structures on GPUs, we develop a message-level parallel implementation that maps batches of message updates to GPU execution.

\subsubsection{Message-Level Parallelization}
We adopt a message-level parallelization strategy, where each GPU thread is responsible for computing a single message on an edge. For each edge $(a, v) \in E$, a thread computes either a variable-to-factor message $\mu_{v\to a}$ or a factor-to-variable message $\mu_{a\to v}$.
To compute a batch of factor-to-variable messages obtained from dependency analysis, we organize the computation into two phases.
This separation avoids read-after-write dependencies and ensures that all messages can be computed fully in parallel based on the precomputed messages.
In the first phase, we compute all required variable-to-factor messages $\mu_{v\to a}$ in a kernel launch. According to \Cref{eq:message-vtoa}, each $\mu_{v\to a}$ is computed as the product of incoming factor-to-variable messages and can be computed independently. We assign one thread per edge to compute these messages in parallel.
In the second phase, we compute the factor-to-variable messages $\mu_{a\to v}$ using BP with local structures. As shown in \Cref{alg:one_m}, each message is evaluated as the sum of multiple sub-messages. Each thread iterates over the sub-messages, checks the indicator conditions, and accumulates the results accordingly.
This design exposes massive parallelism across edges while maintaining good data locality within each thread.

\subsubsection{Memory Organization}
Efficient memory organization is critical for achieving high performance on GPUs, especially for large-scale factor graphs generated by program analysis. In \textsc{FastLBP}, we adopt a memory layout that ensures both compact storage and efficient memory access.
We store all messages in global GPU memory using a CSR format. Specifically, we maintain two separate CSR arrays for variable-to-factor messages and factor-to-variable messages, respectively. 
Each CSR representation consists of two key auxiliary arrays for indexing. One is the row offset array that represents the starting position of each row. The other is the column index array that contains the column indices of the corresponding value.
To facilitate efficient access patterns, we align the CSR rows with the computation direction. When computing a factor-to-variable message $\mu_{a \to v}$, all required variable-to-factor messages associated with factor $a$ are stored contiguously by using factor identifiers as row indices. 
Conversely, when computing $\mu_{v \to a}$, we organize messages using variable identifiers as row indices. 
This layout allows each thread to access all necessary inputs via sequential memory reads, reducing indexing overhead. 

For BP with local structures, instead of storing factor functions as dense tables, we store them as a flattened array of sub-message indicators to effectively support BP with local structures.
We represent the indicator array as $indicator[N][L]$, where $N = |E|$ is the number of edges and $L = \max_{a \in A} l_a$ is the maximum number of terms among all factors.
For efficient GPU storage and access, we organize this array in a CSR-like format using factor identifiers as row indices, so that all indicators associated with the same factor are stored contiguously.
For each edge $e = (a, v)$, $indicator[e]$ is an integer array of length $L$. For $i \leq l_a$, $indicator[e][i]$ follows the original semantics in \Cref{sec:bp-local-structure}. That is, $indicator[e][i] = c$ if $v$ is constrained to value $c$ in the $i$-th term, and $indicator[e][i] = -1$ otherwise. For $i>l_a$, we pad $indicator[e][i]$ with a sentinel value (e.g., -2) indicating an invalid term.
This representation avoids storing factor tables of exponential size and enables each thread to perform BP with local structures, significantly reducing memory footprint and improving computational efficiency.

\subsubsection{Kernel Launch}
To execute message updates efficiently on GPUs, we organize computation into batches derived from dependency analysis. Each batch consists of messages that can be computed independently without synchronization.
For each batch, we launch a GPU kernel where each thread is assigned a single message computation. Each thread is provided with the corresponding factor, variable, and edge identifier. The edge identifier is computed using the CSR representation and determines the location for storing the output message. These identifiers allow each thread to efficiently locate input messages and write results without contention.
Batching not only reduces kernel launch overhead but also improves efficiency by exposing massive parallelism.

\section{Experimental Evaluation}
\label{sec:evaluation}

In this section, we aim to answer the following research questions:
\begin{enumerate}[label=\textbf{RQ\arabic*.}]
\item How effective is \textsc{FastLBP} in improving the efficiency of program analysis?
\item How important is supporting flexible update strategies for program analysis?
\end{enumerate}

We describe our experimental setup in \Cref{sec:experimental-setup}, and answer the above questions in \Cref{sec:efficiency} and \Cref{sec:update-impact}, respectively.

\subsection{Experimental Setup}
\label{sec:experimental-setup}

We implement \textsc{FastLBP} in C++ and CUDA, reusing the frontend of libDAI~\cite{mooij2010libdai} for reading factor graphs. 
We conduct all experiments on a Linux machine with \SI{2.40}{\giga\hertz} Intel Xeon processors, \SI{256}{\giga\byte} RAM, and a single NVIDIA GeForce RTX 4090 GPU. 

\subsubsection{Program Analysis Applications}

We evaluate \textsc{FastLBP} on two probabilistic program analysis
applications, SmartFL~\cite{zeng2022fault} and
Bingo~\cite{raghothaman2018user}, which target fault localization and
alarm prioritization, respectively.

SmartFL performs probabilistic fault localization. It constructs a
factor graph from static analysis information and dynamic execution traces,
and uses BP to infer the posterior probability that each program element is faulty. Based on these probabilities, SmartFL ranks suspicious
program elements to help users inspect likely faulty locations.

Bingo prioritizes alarms produced by Datalog program analysis. In
our evaluation, the underlying analysis is a flow- and context-sensitive datarace analysis~\cite{naik2006effective}. Bingo constructs a
PGM that captures the probabilistic dependencies among analysis facts
and rules, and uses BP to estimate which alarms are likely true. It repeatedly presents the
highest-ranked alarm to the user for inspection and updates the posterior probabilities based on the feedback about whether the inspected alarm corresponds to a true bug.

\subsubsection{Benchmarks} 
We use the same benchmarks as in the previous work~\cite{wu2025belief}, which we describe below.
% Specifically, we evaluate our approach on two probabilistic program analysis
% applications: SmartFL~\cite{zeng2022fault} and Bingo~\cite{raghothaman2018user}. 
% Both systems construct large factor graphs derived from program
% analysis and perform probabilistic inference using BP.

For SmartFL, we use the Defects4J benchmark suite~\cite{just2014defects4j}. 
After filtering invalid cases where the factor graphs cannot be constructed by SmartFL, we
obtain 194 factor graphs. On average, each benchmark instance contains about 102K lines
of code, and the generated factor graph contains approximately 258K
variables and 812K edges.

For Bingo, we use the benchmark suite consisting of eight Java programs collected from previous work~\cite{eslamimehr2014race, zhang2017effective} and the DaCapo benchmark suite~\cite{blackburn2006the}.
% Bingo performs a flow- and context-sensitive datarace analysis~\cite{naik2006effective} on the benchmark.
The average program size is about 81K lines of code. The corresponding factor graphs generated by Bingo contain on average 62K variables and 159K edges.

\subsubsection{Baselines}
We compare \textsc{FastLBP} against two baselines:~\citet{wu2025belief} and PGMax~\cite{zhou2024pgmax}. \citet{wu2025belief} implemented their approach on SmartFL and libDAI, evaluating it on SmartFL and Bingo, respectively. Both implementations are single-threaded. We adopt these implementations as the state-of-the-art CPU-based baselines. 
PGMax~\cite{zhou2024pgmax} is a Python/JAX~\cite{jax2018github} package that implements BP for discrete factor graphs with arbitrary shapes. It only supports synchronous BP with a fixed number of iterations. PGMax can run on GPUs since JAX supports GPU acceleration. We adopt PGMax as the state-of-the-art GPU-based baseline in our evaluation.

\subsection{Efficiency}
\label{sec:efficiency}

\subsubsection{Efficiency on SmartFL} 
To evaluate the efficiency of \textsc{FastLBP}, we compare it with \citet{wu2025belief} and PGMax~\cite{zhou2024pgmax} on SmartFL. Since the original setting of SmartFL adopts synchronous BP, all approaches use synchronous BP in this experiment to ensure a fair comparison. The number of message-passing iterations is set to 500, at which point BP has converged on most factor graphs (190/194). 

Overall, \textsc{FastLBP} significantly outperforms both baselines. \Cref{fig:efficiency-smartfl} shows the inference time against the number of variables in the factor graphs. 
Compared to \citet{wu2025belief}, \textsc{FastLBP} achieves a $17.42\times$ speedup on average (geometric mean). 
Since \citet{wu2025belief} already uses BP with local structures, the speedup primarily comes from GPU parallel execution.
The advantage becomes more pronounced as the size of the factor graph increases. For small graphs (around $10^2$ to $10^3$ variables), the performance differences between them are relatively small due to the limited amount of parallelism available. However, as the number of variables grows beyond $10^4$, \textsc{FastLBP} exhibits significantly lower inference time than \citet{wu2025belief}.

\begin{figure}[t]
    \centering
    \includegraphics[width=\linewidth]{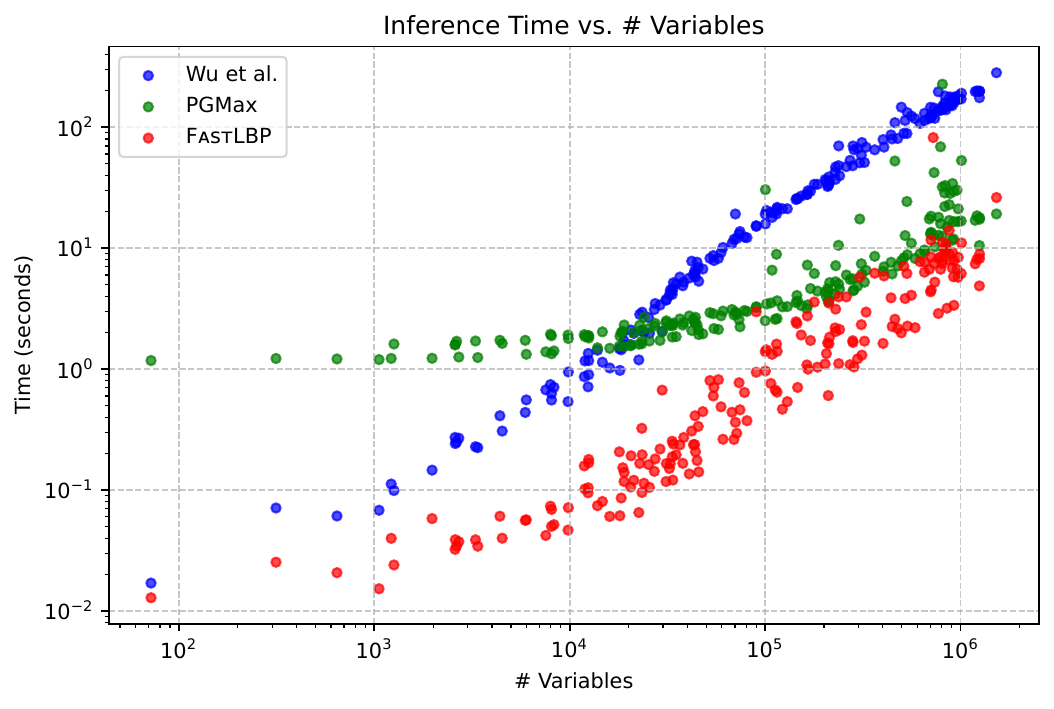}
    \caption{Comparison of Inference Time on SmartFL. Each point corresponds to a test program.}
    \Description{\textsc{FastLBP} outperforms baselines}
    \label{fig:efficiency-smartfl}
\end{figure}

\textsc{FastLBP} also outperforms PGMax, achieving an average speedup of $6.14\times$. 
This improvement highlights the benefit of \textsc{FastLBP}'s efficient implementation of BP with local structures, which PGMax does not support.
The speedup does not increase monotonically with the number of variables because variable count is not the only factor that determines GPU performance. 
One reason is that PGMax incurs JAX compilation overhead, whose relative impact is greater on smaller graphs. Another reason is that graph size affects the GPU utilization of the two approaches differently.
Moreover, \textsc{FastLBP}'s benefit from exploiting local structures depends more directly on factor degree than on the total number of variables. Therefore, \textsc{FastLBP} may benefit more from local structures on a smaller graph with high-degree factors than on a larger graph dominated by low-degree factors.
We note that \texttt{Chart-15} is not shown in \Cref{fig:efficiency-smartfl} because PGMax runs out of memory on this factor graph. This happens because PGMax adopts a tabular representation of factors, while \texttt{Chart-15} contains factors connected to many variables.

We also check the correctness of \textsc{FastLBP} on SmartFL. As shown in \Cref{thm:correctness}, \textsc{FastLBP} is mathematically equivalent to sequential implementations, but minor floating-point errors might still occur due to implementation details. We quantify the differences in marginal probabilities using relative error. Across all test cases, \textsc{FastLBP} has relative errors less than $10^{-8}$ compared to \citet{wu2025belief}, confirming its correctness.

\subsubsection{Efficiency on Bingo} 
% Bingo performs probabilistic inference with user feedback as evidence. In each iteration, it queries the user about the alarm with the highest likelihood, then updates the posterior distribution based on the feedback. Efficient probabilistic inference is therefore crucial for quickly identifying true alarms. We evaluate the overall efficiency of inpecting true alarms.
Since Bingo performs BP before each round of user feedback, efficient inference is crucial for identifying true alarms quickly.
We compare \textsc{FastLBP} with \citet{wu2025belief} on Bingo. 
We adopt the asynchronous BP introduced in \Cref{sec:update-strategy}, which is also the update strategy originally used by Bingo. Empirically, BP does not converge in most factor graphs generated by Bingo, making asynchronous update strategy more suitable for this workload. We do not compare PGMax in this experiment because it only supports synchronous BP. We follow the original stopping criterion, limiting the number of message-passing iterations to between 500 and 1,000 when BP does not converge.

\Cref{fig:efficiency-bingo} shows the number of true alarms inspected over the accumulated inference time for each benchmark.
On most benchmarks, \textsc{FastLBP} identifies all true alarms in significantly less time. The only exception is \texttt{weblech}, where the factor graph is relatively small (with 696 variables). Under the constraints of the asynchronous update strategy, \textsc{FastLBP} cannot achieve high parallelism on such small graphs. The benefits of GPU parallelism do not outweigh the overhead introduced by dependency analysis and memory transfers. As the graph size increases, the advantage of \textsc{FastLBP} becomes increasingly evident. On \texttt{hedc}, \textsc{FastLBP} is slightly faster than \citet{wu2025belief}. On the remaining benchmarks (\texttt{luindex}, \texttt{jspider}, \texttt{avrora}, \texttt{xalan}, \texttt{sunflow}, and \texttt{ftp}), \textsc{FastLBP} runs several times faster. Although \textsc{FastLBP} is marginally slower on small graphs, this has negligible practical impact because the total BP time on such graphs is inherently small (less than 3 seconds on \texttt{weblech}). Thus, the performance difference does not significantly affect user experience. In contrast, users are primarily concerned with execution time on large-scale graphs, where the total inference time can extend to several hours or even exceed one day. It is precisely in these large-scale scenarios that \textsc{FastLBP} excels, significantly accelerating BP and substantially reducing user waiting time.

\begin{figure}[t]
    \centering
    \includegraphics[width=\linewidth]{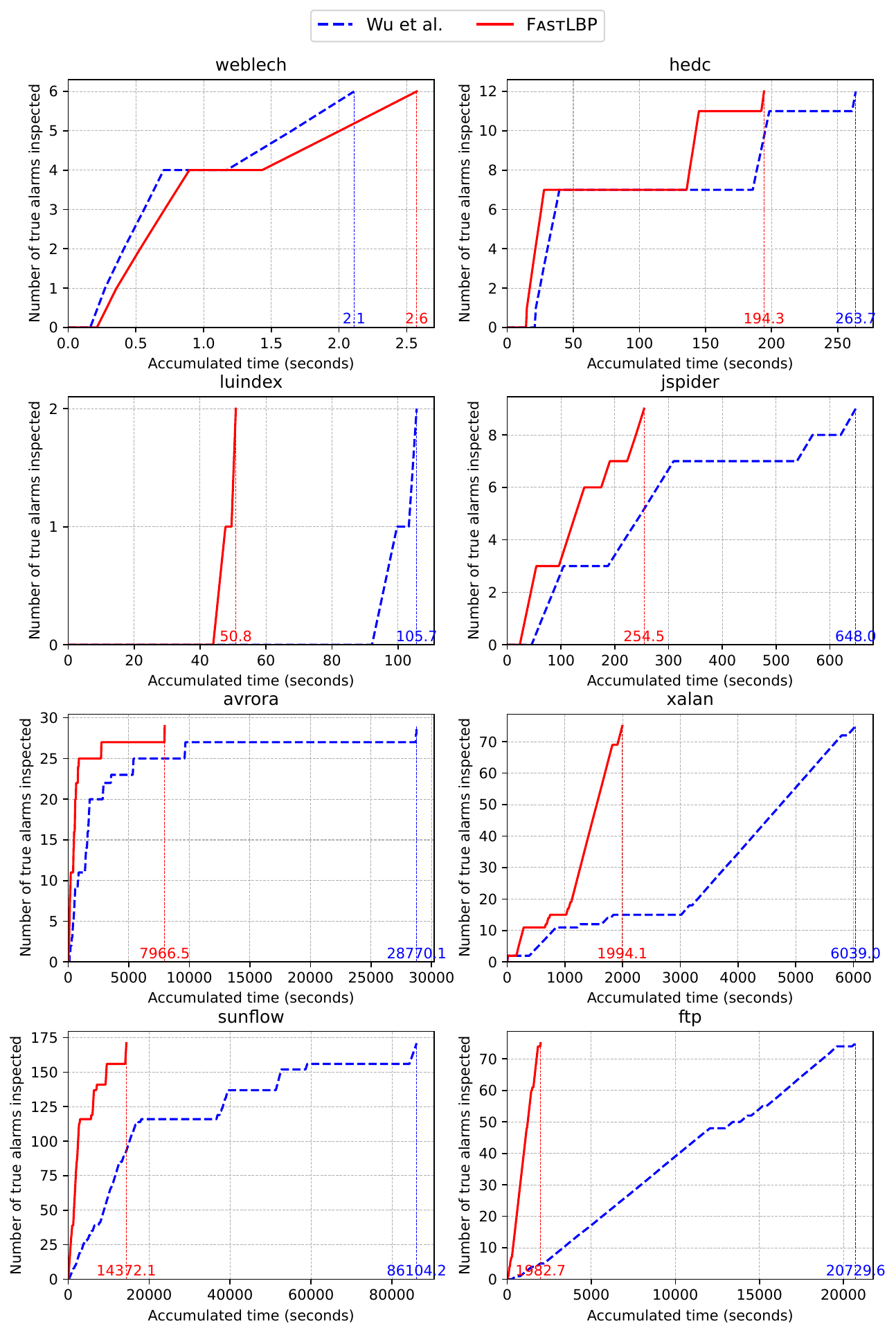}
    \caption{Comparison of Inference Time on Bingo.}
    \Description{\textsc{FastLBP} outperforms baselines}
    \label{fig:efficiency-bingo}
\end{figure}

\textsc{FastLBP} achieves an average speedup of $2.82\times$ in total inference time compared to \citet{wu2025belief}. The relatively moderate overall speedup is mainly due to two reasons. First, the asynchronous update strategy required by Bingo limits the amount of parallelism that can be exploited by GPU execution. Second, several benchmarks contain relatively small factor graphs. When excluding benchmarks where \citet{wu2025belief} finishes within 15 minutes (i.e., \texttt{weblech}, \texttt{hedc}, \texttt{luindex}, and \texttt{jspider}), the average speedup increases to $5.12\times$, with a maximum speedup of $10.46\times$.
These results demonstrate that \textsc{FastLBP} effectively supports asynchronous update strategies required by program analysis applications.

\subsection{The Importance of Supporting Flexible Update Strategies}
\label{sec:update-impact}
To investigate the importance of supporting flexible update strategies in program analysis, we evaluate how different BP update strategies affect the quality of inference results on Bingo. In particular, we compare asynchronous BP used by \textsc{FastLBP} and \citet{wu2025belief} with synchronous BP implemented in PGMax.
We do not evaluate on SmartFL because synchronous BP already converges reliably on SmartFL. We use the default stopping criterion of Bingo for \textsc{FastLBP} and \citet{wu2025belief}, which stops at 500 to 1,000 iterations based on the convergence rate. 
Since PGMax runs BP for a fixed number of iterations, we set the number of iterations to 1,000 to match the upper bound of Bingo's stopping criterion.

Following Bingo, we employ three metrics to quantify the quality of inference results: inversion count, Rank-100\%-T, and Rank-90\%-T. An inversion is defined as a pair of alarms where a false alarm is ranked higher than a true alarm in the user's inspection sequence. Formally, let $l_1, \ldots, l_n \in \{0, 1\}$ denote the ground-truth labels of the alarms provided by the user. Then, the inversion count is calculated as: $\text{Inversion}(l_1, \dots, l_n)=\sum_{i=1}^n (1 - l_i) \sum_{j=i+1}^n l_j$. Rank-100\%-T and Rank-90\%-T represent the numbers of interactions to identify all and 90\% of the true alarms. 
Lower values indicate that users can prioritize and inspect true alarms more rapidly. Consequently, higher accuracy in BP results correlates with improved alarm ranking and reduced inspection effort.

\Cref{tab:update-impact} shows the results of the three approaches.
Compared with \citet{wu2025belief}, \textsc{FastLBP} shows only negligible differences in all three metrics. Theoretically, \textsc{FastLBP} should produce identical results to the sequential implementation under the same update strategy. The small discrepancies observed in practice are primarily due to floating-point errors between the two implementations. This leads to differences in the marginal probabilities, thus affecting the final alarm rankings. These results confirm that \textsc{FastLBP} preserves the inference quality of the sequential implementation while enabling GPU acceleration.

\begin{table*}[t]
\caption{Results of \textsc{FastLBP} and PGMax compared to \citet{wu2025belief}. PGMax runs out of memory (OOM) on sunflow. Average Change represents the geometric mean of the ratios, minus one.}
\label{tab:update-impact}
\centering
\begin{tabular}{@{}cccccccccc@{}}
\toprule[1.5pt]
\multirow{2}{*}{Program} & \multicolumn{3}{c}{Inversion} & \multicolumn{3}{c}{Rank-100\%-T} & \multicolumn{3}{c}{Rank-90\%-T} \\
\cline{2-10}
& \citet{wu2025belief} & \textsc{FastLBP} & PGMax & \citet{wu2025belief} & \textsc{FastLBP} & PGMax & \citet{wu2025belief} & \textsc{FastLBP} & PGMax \\
\midrule[1pt]
weblech & 18 & 18 & 19 & 11 & 11 & 11 & 11 & 11 & 11 \\
hedc & 347 & 347 & 1109 & 87 & 87 & 152 & 65 & 65 & 151 \\
luindex & 29 & 29 & 120 & 17 & 17 & 62 & 17 & 17 & 62 \\
jspider & 72 & 64 & 398 & 26 & 20 & 196 & 26 & 20 & 196 \\
avrora & 3,187 & 3,048 & 5,625 & 900 & 885 & 926 & 293 & 302 & 381 \\
xalan & 3,110 & 3,890 & 6,306 & 126 & 141 & 868 & 117 & 128 & 122 \\
sunflow & 25,200 & 25,497 & OOM & 829 & 825 & OOM & 595 & 566 & OOM \\
ftp & 570 & 606 & 14,999 & 92 & 95 & 430 & 81 & 81 & 304 \\
\midrule[1pt]
Average Change & - & $1.69\%\uparrow$ & $256.14\%\uparrow$ & - & $1.73\%\downarrow$ & $186.68\%\uparrow$ & - & $2.37\%\downarrow$ & $128.48\%\uparrow$ \\
\bottomrule[1.5pt]
\end{tabular}
\end{table*}

In contrast, PGMax exhibits significantly worse results across all metrics. The inversion count, Rank-100\%-T, and Rank-90\%-T increase by $256.14\%$, $186.68\%$, and $128.48\%$, respectively. This degradation significantly affects user experience because users must inspect many more false alarms before discovering true bugs. The main reason is that PGMax only supports synchronous BP. In program analysis frameworks such as Bingo, BP often fails to converge and tends to exhibit more oscillating behaviors under synchronous BP than asynchronous BP. 

These results highlight the importance of supporting flexible update strategies for program analysis. In particular, asynchronous BP is essential for workloads where synchronous BP leads to unstable inference results. By supporting both synchronous and asynchronous update strategies, \textsc{FastLBP} can adapt to different workloads while maintaining both efficiency and accuracy.

\section{Related Work}
\label{sec:related-work}

Our work is related to research on (1) GPU acceleration for BP, (2) parallel BP on CPUs, and (3) probabilistic program analysis using BP. We discuss the related work below.

\paragraph{GPU Acceleration for BP}
GPU acceleration of BP has been explored in several domain-specific applications, such as computer vision~\cite{grauergray2008gpu, brunton2006belief,liang2009hardware} and channel coding~\cite{reddy2012a,abburi2011a,romero2012sequential}. These approaches design GPU implementations tailored to the characteristics of PGMs in specific domains. 
In contrast, our work focuses on accelerating BP in the context of program analysis, where the resulting PGMs are significantly larger in scale and more complex in graph structure.
Beyond domain-specific applications, \citet{zhou2024pgmax} implement BP using JAX~\cite{jax2018github}, enabling GPU acceleration for general PGMs. Compared to their work, our approach supports more sophisticated update strategies and a more efficient BP algorithm.
Randomized Belief Propagation~\cite{merwe2019message} randomly selects a subset of messages to update in each iteration, effectively balancing convergence behavior and parallel efficiency.
In contrast, our approach supports a broader range of update strategies and can be easily adapted to incorporate their update strategy.

\paragraph{Parallel BP on CPUs} 
To improve convergence, prior work on CPU-based parallel BP has focused on designing specific update strategies -- frontier-based update strategies~\cite{gonzalez2009residual,aksenov2020relaxed}. They design update strategies that use the concept of residuals~\cite{elidan2006residual} and update the most influential messages during each iteration to balance parallelism and convergence behavior.
While effective on CPUs, these approaches incur significant overhead on GPUs~\cite{merwe2019message}, limiting their applicability in GPU settings. 
In comparison, we propose a framework that enables flexible update strategies through a unified representation and dependency analysis, capturing diverse user-specified dependencies. Our approach can also be naturally extended to CPU execution.
We do not directly compare with CPU-based parallel BP approaches because our target workloads expose massive fine-grained parallelism that is better matched to GPU execution. 
Moreover, the per-update scheduling overhead of frontier-based approaches can become prohibitive for large factor graphs.

\paragraph{Probabilistic Program Analysis Using BP} 
Recent years have seen increasing interest in applying probabilistic reasoning to program analysis. In particular, Bayesian program analysis transforms static analysis derivation graphs into PGMs and reasons about the likelihood that a reported alarm corresponds to a real bug. 
Prior work in this area can be classified in four directions. 
The first line of work focuses on incorporating diverse forms of posterior information to refine the alarm ranking. Posterior information includes user feedback on the alarms~\cite{raghothaman2018user}, differences between the derivations of alarms~\cite{heo2019continuously}, test execution results~\cite{chen2021boosting}, and informal information in the program~\cite{li2025combining}. 
The second line explores how to refine Bayesian networks to mitigate false generalization of posterior information. \citet{kim2022learning} proposed an approach to learn new rules and probabilities. \citet{zhang2024learning} and \citet{shi2025on} aimed at learning to select a suitable abstraction.
In the third line, \citet{lin2026beer} studied alarm selection during interactive Bayesian program analysis and proposed an approach that balances exploration and exploitation instead of always selecting the alarm with the highest posterior probability.
The last line applies probabilistic program analysis to augment other software engineering tasks, such as fault localization~\cite{zeng2022fault, wu2025smartfl} and fuzzing~\cite{zhang2026fuzzing}. 
These approaches typically model program semantics as PGMs and perform BP for marginal inference. However, the resulting models are often large and complex, making BP a major performance bottleneck. Our work can accelerate these systems by providing an efficient and accurate inference library.

\section{Conclusion}
\label{sec:conclusion}

We present \textsc{FastLBP}, a GPU-accelerated belief propagation framework for program analysis.
By introducing a unified representation for update strategies and a dependency analysis algorithm, \textsc{FastLBP} enables parallel execution for flexible update strategies.
Additionally, our memory-efficient GPU implementation of BP with local structures effectively handles logical constraints in program analysis.
Our evaluation on SmartFL and Bingo demonstrates that \textsc{FastLBP} significantly outperforms both CPU-based and GPU-based approaches in efficiency without sacrificing accuracy.

\begin{acks}
  We thank the anonymous reviewers for their constructive feedback.
This work was partly supported by the National Natural Science Foundation of China (NSFC, Grant No. W2411051).
\end{acks}

\section*{Data Availability Statement}
The artifact for this paper is available at Zenodo~\cite{artifact}. It includes source code, scripts, and data required to reproduce the results presented in \Cref{sec:evaluation}.

\balance
\bibliographystyle{ACM-Reference-Format}
\bibliography{sample-base}

% \appendix

\end{document}